\documentclass[12pt,english,eqsecnum,floats,prd,aps,nofootinbib,superscriptaddress,amsfonts,longbibliography,tightenlines,preprint]{revtex4-1}
\usepackage[T1]{fontenc}
\usepackage[latin9]{inputenc}
\usepackage{babel}
\usepackage{amsmath}
\usepackage{amsthm}
\usepackage{amssymb}
\usepackage[all]{xy}
\usepackage[unicode=true,pdfusetitle,
 bookmarks=true,bookmarksnumbered=false,bookmarksopen=false,
 breaklinks=false,pdfborder={0 0 1},backref=false,colorlinks=false]
 {hyperref}
\usepackage{breakurl}

\makeatletter
\theoremstyle{plain}
\newtheorem{thm}{\protect\theoremname}[section]
\theoremstyle{assumptionus}
\newtheorem{assumption}[thm]{\protect\assumptionusname}
\theoremstyle{conjectureus}
\newtheorem{conjecture}[thm]{\protect\conjectureusname}
\theoremstyle{observationus}
\newtheorem{ob}[thm]{\protect\observationusname}
\theoremstyle{remarkus}
\newtheorem{bem}[thm]{\protect\remarkusname}
\theoremstyle{lemmaus}
\newtheorem{lemma}[thm]{\protect\lemmausname}
\theoremstyle{conclusionus}
\newtheorem{conclusion}[thm]{\protect\conclusionusname}
\theoremstyle{theoremus}
\newtheorem{theo}[thm]{\protect\theoremusname}
\theoremstyle{definitionus}
\newtheorem{defi}[thm]{\protect\definitionusname}
\theoremstyle{corollaryus}
\newtheorem{koro}[thm]{\protect\corollaryusname}

\usepackage[all]{xy}
\allowdisplaybreaks[1]

\makeatother

\providecommand{\assumptionusname}{Assumption}
\providecommand{\conclusionusname}{Conclusion}
\providecommand{\conjectureusname}{Conjecture}
\providecommand{\corollaryusname}{Corollary}
\providecommand{\definitionusname}{Definition}
\providecommand{\lemmausname}{Lemma}
\providecommand{\observationusname}{Observation}
\providecommand{\remarkusname}{Remark}
\providecommand{\theoremname}{Theorem}
\providecommand{\theoremusname}{Theorem}

\begin{document}

\title{Emergent Space-Time via a Geometric Renormalization Method}

\author{Saeed Rastgoo}
\email{saeed@xanum.uam.mx}

\affiliation{Departamento de F\'{i}sica, Universidad Aut\'{o}noma Metropolitana
- Iztapalapa~\\
San Rafael Atlixco 186, Mexico D.F. 09340, Mexico}

\author{Manfred Requardt}
\email{requardt@theorie.physik.uni-goettingen.de}

\affiliation{Institut f\"{u}r Theoretische Physik Universitaet Goettingen~\\
Friedrich-Hund-Platz 1, 37077 Goettingen, Germany}

\date{\today}
\begin{abstract}
We present a purely geometric renormalization scheme for metric spaces
(including uncolored graphs), which consists of a coarse graining
and a rescaling operation on such spaces. The coarse graining is based
on the concept of quasi-isometry, which yields a sequence of discrete
coarse grained spaces each having a continuum limit under the rescaling
operation. We provide criteria under which such sequences do converge
within a superspace of metric spaces, or may constitute the basin
of attraction of a common continuum limit, which hopefully, may represent
our space-time continuum.

We discuss some of the properties of these coarse grained spaces as
well as their continuum limits, such as scale invariance and metric
similarity, and show that different layers of spacetime can carry
different distance functions while being homeomorphic. 

Important tools in this analysis are the Gromov-Hausdorff distance
functional for general metric spaces and the growth degree of graphs
or networks. The whole construction is in the spirit of the Wilsonian
renormalization group.

Furthermore we introduce a physically relevant notion of dimension
on the spaces of interest in our analysis, which e.g. for regular
lattices reduces to the ordinary lattice dimension. We show that this
dimension is stable under the proposed coarse graining procedure as
long as the latter is sufficiently local, i.e. quasi-isometric, and
discuss the conditions under which this dimension is an integer. We
comment on the possibility that the limit space may turn out to be
fractal in case the dimension is non-integer. At the end of the paper
we briefly mention the possibility that our network carries a translocal
far-order which leads to the concept of wormhole spaces and a scale
dependent dimension if the coarse graining procedure is no longer
local.
\end{abstract}
\maketitle
\tableofcontents{}

\section{Introduction}

Classical macroscopic space-time (henceforth: S-T) is a continuous
manifold on macroscopic scales. The physical (quantum) fields live
on this continuous manifold as separate entities. Macroscopic objects
move freely through S-T. On the other hand, on a more microscopic
(but, compared to the infamous Planck scale, still mesoscopic) scale
the quantum vacuum appears to be full of quantum fluctuations while
on a still finer scale even S-T itself is expected to wildly fluctuate.

While in string theory the framework is (at least initially) constructed
over smooth (higher dimensional) manifolds, in most of the other approaches
to quantum gravity one assumes that at a primordial level S-T is both
discrete and presumably quite erratic. A crucial concept in these
latter approaches is the notion of \emph{background independence}.

There have been various attempts to introduce a discrete structure
and reconstruct smooth classical (or macroscopic) S-T from such a discrete and irregular
substratum (we will discuss in more detail in the following section what is meant by these attributes). But as this paper is not intended to be a review we are able to mention
only a few randomly selected sources. Examples are the \emph{spin
networks} and \emph{spin foam} models in loop quantum gravity (LQG) and its
path integral versions. Some more recent examples are the \emph{quantum
graphity} approach \citep{Thiemann,Rov1,Baez,Smo1,Marko1,Smo2,Oriti1,Perez,Konopka1,Konopka2,Dittrich1409,Dittrich1609}, 
the \emph{group field theory} framework and the \emph{random tensor
networks} \citep{Oriti2,Chen}. A nice description of the whole field can for example be found in \citep{Oriti1302}
with emphasis on the emergence of space-time in the various approaches.

All these approaches are, to a larger or lesser degree, derived from attempts
to directly quantize classical general relativity (GR). The same holds
for  \emph{dynamical triangulation} and related frameworks
like e.g. \emph{Regge calculus} \citep{Ambjorn}. Another approach
has been developed by ourselves and coworkers and is based on generalizations
of cellular automata (CA), i.e. \emph{Structurally Dynamic Cellular
Networks} (SDCN); for more information see the recent review \citep{RequRast}. To this class also belongs e.g. \cite{trugenberger1610} which studies random-Ising-like models of space-time. One should also mention the work in general network theory, see e.g. the recent \citep{Bianconi1509}
or our paper \cite{scalefree}

The point of view, shared by all these approaches is the conviction
that S-T on its most primordial level is a dynamic substratum consisting
of certain elementary degrees of freedom. Their nature, however, may
be different in the various schools. In many frameworks they carry a certain 
a priori geometric flavor (inspired by the \emph{simplicial resolution} of
continuous manifolds) and consist for example of infinitesimal triangles or tetrahedrons  having, a fortiori, edges of a certain infinitesimal length.
In other approaches these microscopic degrees of freedom (DoF) are
viewed more abstractly as elementary cells, carrying internal states,
interacting with each other (or exchanging information) via elementary
interactions. In this latter case the quantum vacuum is regarded as
a huge irregular dynamical system (e.g. the SDCN). The geometric notions
are here considered to emerge from a primarily non-geometric substratum
in the spirit of J. A. Wheeler: Geometry from Non-Geometry (see Sect. 44.4 of \citep{Phonebook}).
In e.g. LQG the elementary DoF carry both qualities to a greater or lesser degree, i.e. geometric ones and more abstract ones. Another important point: the DoF may not even have any local character. The localization of objects on coarser scales may  also be emergent in a relational way.

In concluding our brief discussion of the various points of view expressed
in the frameworks addressed above we want to comment on some remarks
made in the nice contribution of Bombelli et. al. \citep{Bombelli}
because it offers us the opportunity to clarify some points of principle
and misunderstandings. This paper deals mostly with the reconstruction
of a continuum manifold via a \emph{piecewise linear} (embedded) manifold
which is, on its side, derived from certain graphs. This they call
the \emph{inverse problem}.

One should say that this is a fairly widespread strategy in the mathematics
of manifolds with a huge amount of published results. An important
question in fundamental physics is to what extent nature on its most
primordial level is actually concerned with such geometric micro objects
like simplices, tetrahedra and the like.

We shall argue in the following that quantum space-time is, in our
view, rather an extremely complex and erratic dynamical system, consisting
of an array of elementary DoF together with elementary interaction
among these DoF. It is then the task to derive geometric notions (and
for example continuum analysis) from such a primordial substratum.
We developed such concepts in e.g. \citep{Requ1} as well as a certain
\emph{discrete calculus} (which has relations to \emph{non-commutative
geometry}). We showed for example that one can develop a kind of \emph{(co)homology
theory} (see section 3.2) by associating simplices to subsets of DoF
with elementary interactions existing between all the respective pairs
of DoF in the subset. For more information see our recent \citep{RequRast}.

To clear things up a bit regarding some of the comments in \citep{Bombelli}
about the notion of the dimension introduced in \citep{Requ2}, we
mention a few words. We developed this concept of dimension for graphs
and networks in \citep{Requ3}, being mainly motivated by physical
ideas. We studied how space dimension really enters in the physical
formulas in say statistical mechanics, critical phenomena etc. We
realized that what typically really matters is the number of new interaction
partners a local site sees after consecutive steps on e.g. a lattice,
embedded in some continuous space. It then happens that the dimension
of the embedding space enters in a characteristic way in the physical
formulas. We observed later in \citep{Requ2} that a related notion
was used in a beautiful field of pure mathematics, that is, \emph{geometric
group theory, viz. Cayley graphs} and is there called the \emph{growth
degree}.

One should furthermore mention that our concept of dimension has,
despite its superficial similarity, nothing to do with any \emph{fractal
dimension}. The latter concept describes the behavior of a system
in the infinitely small while our \emph{(scaling) dimension} rather
characterizes the large scale properties of graphs and networks and
is an important invariant of such structures. Its advantage is that
it shows that integer dimensions are very particular while non-integer
dimensions are rather the rule (as in the fractal case). This gives
us a tool and criterion to single out spaces having an integer macro-dimension
(like our own S-T). On the other hand, many of the more geometric
approaches are dealing right from the start with integer dimensions.
We think, it may be interesting to learn that the integer dimensionality
of our physical S-T needs some explanation.

As we already indicated, our aim in the following is it to develop a coarse graining scheme
and continuum limit comprising as many different frameworks as possible. One should note, however, that this is an ambitious task as the meaning of coarse graining and/or continuum limit may have a different meaning in the various approaches. The same holds for the concept of dynamics (we comment on these points in the next section). 
To this end we make the following assumption:
\begin{assumption}
We assume that S-T on its most primordial level is a complex dynamical
system, consisting of a huge array of microscopic DoF together with
a (random) distribution of elementary interactions connecting these
DoF.
\end{assumption}
One should perhaps emphasize that working with networks or graphs
does \emph{not} mean that we really think of the most primordial objects
as extensionless points, quite to the contrary, they are usually assumed
to represent certain lumps having an internal structure which, on
the respective scale of resolution, cannot be resolved (cf. our paper
\citep{Menger} and the remarks by Menger at the end of his contribution
in \citep{Einstein}). Our following contribution shows that Einstein's
skeptical remarks concerning such a radical program made in \citep{Einstein}
was perhaps too conservative.

As a further remark, there is yet another approach towards a space
of spaces in (quantum) space-time physics. i.e. classifying space
structure via the spectrum of the Laplacian \citep{Seriu}, but we
do not touch this field in this paper.

The structure of the paper is as following: we start in section \ref{sec:Main-strategy}
by presenting the big picture and main steps of the work, so that
the reader does not get lost in the details of the following sections.
In section \ref{sec:typical}, we present the idea of a phase cell
in state space of a system, that can be considered as the basin of
attraction of the evolution map. This will be used to categorize the
class of states (discrete spaces) that yield the same smooth continuum
limit. Section \ref{sec:Graphs-mspace} contains necessary concepts
about the graphs as metric spaces. In section \ref{sec:coarse} we
present our generic coarse graining schemes, consisting of quasi-isometry
and rough isometry, and provides some examples of each. In section
\ref{sec:supermetric space}, we introduce the Gromov-Hausdorff space,
and its associated metric, the Gromov-Hausdorff metric with respect
to which (non)isometry of spaces is measured. Sections \ref{sec:Converg-1}
and \ref{sec:Converg-2} are devoted to the notions of convergence
of a sequence of metric spaces, and defining their rescaling and continuum
limit. In section \ref{sec:Geo-Renorm}, we combine all the information
in previous sections and fully develop and describe the geometric
renormalization process that can lead from a discrete structure like
a graph to a continuum limit such as a manifold. In section \ref{Dim}
we present a notion of dimension in our spaces, and very briefly discuss
its properties under the geometric renormalization process. Finally,
we summarize and make some concluding remarks in section \ref{sec:Conclusion}.
The appendix includes several definitions that are needed for the
paper to be self-consistent.

\section{Main strategy: the big picture\label{sec:Main-strategy}}

Any theory with a discrete or quantum pregeometry, is faced with the
challenge of deriving a continuum limit for this discrete structure,
that looks like a desired smooth spacetime manifold. Of course we
are aware of the facts that discrete does not always mean quantum,
and the concepts of semiclassical limit and continuum limit are, in most cases, not
 the same,  describe different physics and do not commute. 
\begin{bem} We would like to clarify what we mean by a continuum limit. In various approaches it means simply a way of embedding a discrete (approximate) structure into some (preexisting) background manifold as in so called piece-wise linear geometry. Our enterprise is much more ambitious. We rather perform a true scaling limit starting from a sequence of discrete spaces. That is, the limit space can have a quite complicated continuous structure.
\end{bem}

The semiclassical limit is generally associated to large quantum numbers, examples of which can be found in ferromagnetism and superconductivity, which in some cases may be seen as a manifestation of the correspondence principle. In loop quantum gravity for example, semiclassical typically refers to a limit of a fixed finite (generally small) number of large quantum numbers (spins) corresponding to a graph on which a coherent state (labelled by the aforementioned spins) lives. This semiclassical limit on large scales, evidently   corresponds to the Regge-gravity.

The continuum limit on the other hand may be seen as being related to the so called thermodynamic limit, where there is a huge (or infinite) number of DoF  each with  small associated quantum numbers. In LQG, for example, the continuum limit  generally refers to many, perhaps infinite, number of small spins.    

Also there may be states in the theory that
may not have any semiclassical counterpart, such as the replacement
of cosmological singularities with highly quantum nonsingular structures
in loop quantum cosmology \citep{Ashtekar-status-LQC}. 

There are many viewpoints and methods that deal with the issue of getting spacetime (or discrete/continuous general relativity) as their semiclassical or continuum limit. In what follows we try to make some clarifying remarks regarding our method and its relation to some of other ones. 
 We begin with some brief remarks how a continuum limit is  understood in  LQG as presented in  \citep{Dittrich1409,Dittrich1609} 
which refers there to \cite{Thiemann}.
There, the ultimate task is to construct  a continuum physical Hilbert space, satisfying the constraints of the theory, as a limit of an inductive process, associated to iterated refinements of graphs embedded in a spatial manifold. 

What we have in mind, on the other hand, is rather a more direct way of constructing a kind of a macroscopic continuous smooth manifold, representing the classical spacetime, out of a discrete fundamental structure (a graph).  This transition from a graph to a smooth structure is implemented with the help of two different kinds of operations. One of these operations  is, in a sense, related to the coarse graining within discrete structures. The  second one, i.e. a rescaling operation, leads to a  macroscopic continuum (in the optimum case).  It defines the true transition from discrete spaces to continuous spaces and is quite non-trivial.

Another point to be briefly discussed is the kind of structures being  constructed over this discrete network. In canonical LQG, the spin networks, constituting a complete basis of the Hilbert space of the solutions to the quantum Gauss constraint, are represented by graphs whose edges are colored by the irreducible representations of a compact  group ($SU(2)$), and whose vertices are the intertwiners of the representations of the edges that are connected to those vertices.    In our own approach,  inspired by the generalization of cellular automata, the elementary  structure also consists  of vertices with simple internal states (and like LQG obey some kind of consistency relation based on the states of the edges connected to them), and edges which can change their orientation (\emph{directed graphs}) or be deleted or created due to  the dynamics. A choice of the dynamics can be given by an interaction of edge and vertex states yielding a new network, after each evolution step called a ``clocktime'' step. As in the spin network case we can define Hilbert spaces over the vertices and edges and then graph operators such as discrete Laplacians and Dirac operators. This procedure also establishes a connection to Connes' \emph{noncommutative geometry} (for details see for example \cite{Requ1}, \cite{Requlumps} or \citep{Requardt0001}). In section 4 of \cite{RG} we showed that the evolving dynamics belongs to the same general class of graph transformations or dynamics, as is the case for \emph{spin network dynamics} or \emph{causal set dynamics}. I.e., we discuss in this paper primarily the continuum limit of the ``space-like slices'' in  $S-T$. The dynamics is formulated as for spin-networks or causal sets by having a consecutive sequence of such geometric networks states given by some graph transformation rule. For convenience we prefer to call our underlying substratum $S-T$ while, for the time being, we mainly deal with the $S$-part. 

Furthermore, in 
\citep{Requardt0611}
we performed large scale numerical computer studies of various dynamical or statistical parameters of our networks to learn  about their large scale or longtime behavior, or to find \emph{phase transitions}.
That is, we think all these various approaches, while being different with respect to their technical details, are on the other hand sufficiently related so that a joint treatment seems to be reasonable. We would like to resume what we are going to do and what we will not do:
\begin{bem} In the following we will abstract completely from the internal edge and vertex states, the Hilbert space structure and the operators living on the network as well as the dynamics, as we suspect that their respective macroscopic or continuum limit can be studied separately. We concentrate exclusively on the geometric state of our network, that is, its wiring diagram. We surmise that this purely geometric substratum and its evolution (as in the causal set approach) may describe the fine structure of our macroscopic  space-time manifold.
\end{bem}

It is clear that the network dynamics also changes after each coarse graining or renormalization step. In previous work (see e.g. \cite{RG} or our recent review \cite{RequRast}) we formulated for example a class of graph dynamics via an interaction of vertex and edge states which can be applied on each level of coarse graining. By the same token one can study the limit behavior of operators and fields defined on our networks but we will postpone this investigation to future work. By the way, one has to be prepared to deal with operators on fractal spaces which is no easy task.

A last remark concerning the generality of our approach. We think that the definition of coarse graining and/or renormalization via general quasi-isometries comprises many of the possible renormalization schemes, but we have not investigated every suggestion in detail.

We want to conclude this brief survey with a remark concerning the \emph{diffeomorphism invariance} which is a central theme in the theory of gravitation and consequently in LQG. Without constantly mentioning it, it is also encoded in our own approach which is based on the theory of general metric spaces.
\begin{bem} In the following we usually deal with equivalence classes of metrically isomorphic spaces. That is, the limits we are going to construct consist of whole classes of isomorphic spaces which are treated as essentially undistinguishable. In more technical terms (to be defined in the following) they have Gromov-Hausdoff-distance zero. However, there may exist spaces with GH-distance zero, which are not isomorphic. That is, this class may be actually larger. In any case, this kind of equivalence could in our view be considered as being similar to  diffeomorphism invariance. 
\end{bem}

Let us now start to develop the general framework of our approach (the main ideas and a large part of the work
introduced here is based on \citep{Requ2}). As we mentioned earlier, this is a rather direct, and
to a great extent, general approach, since it deals with general types
of graphs, and does not assume any specific Hamiltonian, action,
etc., for the semiclassical or continuous theory, although the final goal is derive an emergent smooth manifold equipped with a metric that satisfied general relativity. More specifically, we ask: given
a graph $G_{0}$ (viz., the geometric abstraction of a network of
interacting sites) as our pregeometric structure and a physically
motivated coarse graining process, is it possible to get to a smooth
manifold as the semiclassical limit of $G_{0}$? and if so, how? and
under what conditions?

This of course depends on the coarse graining process, and on the
criterion used to measure the convergence to a limit space which can
possibly be associated to our space-time manifold. All of this will
be discussed in detail in the following sections. To give a big picture
of the strategy, here is how we attack the problem: 
\begin{enumerate}
\item We consider the superspace $\mathcal{S}$ of all non-compact but locally
compact metric spaces, including suitable uncolored graphs\footnote{The treatment of colored graphs is postponed to future works.}
and suitable manifolds with a metric. This can be seen to be an analogue of the ``theory space'' in renormalization methods.  
\item A graph $\left(G_{0},d_{0}\right)\in\mathcal{S}$, with $d_{0}$ being
a graph metric, is chosen as our initial system, representing the
fundamental layer that constitutes the fundamental discrete level
of the smooth spacetime. To get to the smooth spacetime, two types
of operations are introduced: 
\begin{enumerate}
\item A \textendash to some extent\textendash{} generic coarse graining
scheme $\mathcal{K}$, based on the notion of quasi-isometry. This
coarse graining procedure, assigns to a metric space (in this case
a graph) $\left(G_{1},d_{1}\right)\in\mathcal{S}$, another metric
space $\left(G_{2},d_{2}\right)\in\mathcal{S}$ such that $\left(G_{2},d_{2}\right)$
is the coarse grained space derived from $\left(G_{1},d_{1}\right)$
by applying $\mathcal{K}$ once: 
\begin{equation}
\left(G_{1},d_{1}\right)\stackrel{\mathcal{K}}{\longrightarrow}\left(G_{2},d_{2}\right).
\end{equation}
This process can be performed repeatedly to give us a sequence of
(possibly countably infinite number of ) coarse grained spaces $\left\{ \left(G_{i},d_{i}\right)\right\} ,i=0,\ldots,N$
such that 
\begin{equation}
\left(G_{0},d_{0}\right)\stackrel{\mathcal{K}}{\longrightarrow}\left(G_{1},d_{1}\right)\stackrel{\mathcal{K}}{\longrightarrow}\left(G_{2},d_{2}\right)\stackrel{\mathcal{K}}{\longrightarrow}\cdots\stackrel{\mathcal{K}}{\longrightarrow}\left(G_{N},d_{N}\right),
\end{equation}
where all the $\left(G_{i},d_{i}\right)$'s belong to $\mathcal{S}$.
We call this chain, a coarse graining chain.
\item A rescaling map
\begin{equation}
\phi_{\lambda}:(G_{i},d_{i})\longmapsto(G_{i},\lambda d_{i})
\end{equation}
on each member of the coarse graining chain, such that the limit $\lambda\to0$
of the above map corresponds to the continuum limit $\left(G_{i,\infty},d_{i,\infty}\right)$
of $G_{i}$,
\begin{equation}
\lim_{\lambda\to0}\phi_{\lambda}\left((G_{i},d_{i})\right)=\left(G_{i,\infty},d_{i,\infty}\right).
\end{equation}
\end{enumerate}
\item Combining the two above operations, the following picture arises
\[
\xymatrix{[G_{0}]_{\infty} & [\mathcal{K}(G_{0})]_{\infty} & [\mathcal{K}^{2}(G_{0})]_{\infty} & [\ldots]_{\infty} & [\mathcal{K}^{m}(G_{0})]_{\infty} & [\cdots]_{\infty} & [\mathcal{K}^{n}(G_{0})]_{\infty}\\
G_{0}\ar[r]^{\mathcal{K}}\ar[u]|-{\phi_{\lambda\rightarrow0}} & \mathcal{K}(G_{0})\ar[r]^{\mathcal{K}}\ar[u]|-{\phi_{\lambda\rightarrow0}} & \mathcal{K}^{2}(G_{0})\ar[r]^{\mathcal{K}}\ar[u]|-{\phi_{\lambda\rightarrow0}} & \cdots\ar[r]^{\mathcal{K}}\ar[u]|-{\phi_{\lambda\rightarrow0}} & \mathcal{K}^{m}(G_{0})\ar[r]^{\mathcal{K}}\ar[u]|-{\phi_{\lambda\rightarrow0}} & \cdots\ar[r]^{\mathcal{K}}\ar[u]|-{\phi_{\lambda\rightarrow0}} & \mathcal{K}^{n}(G_{0})\ar[u]|-{\phi_{\lambda\rightarrow0}}
}
\]
where $[\mathcal{K}^{i}(G_{0})]_{\infty}=G_{i,\infty}$. In the coarse
graining chain of discrete spaces (the lower horizontal chain), the
coarse graining operation $\mathcal{K}$ is applied consecutively
until one (or possibly both) of the followings happens:
\begin{enumerate}
\item At a certain point the spaces become roughly isometric (that is, $\mathcal{K}$
makes a transition from being a true quasi-isometry to a rough isometry).
In this case we show that the continuum limits of this row of roughly
isometric spaces are the same continuous space.
\item The coarse graining chain ends in a fixed point or a set of accumulation
points under $\mathcal{K}$. This happens after a finite number of
steps. We think however that this is a non-generic situation. For
more remarks see section \ref{sec:Geo-Renorm}.
\end{enumerate}
\item Before becoming roughly isometric, the memebers of the coarse graining
chain are purely quasi-isometric. In this case we show that the scaling
limits $G_{i,\infty},G_{j,\infty}$ of two such spaces $G_{i}$ and
$G_{j}$ are homeomorphic (even can be chosen to be the same topological
space) but carry different metrics. This implies that different levels
of spacetime can have different metric even if they are the same set.
\item Finally we define a notion of dimension, and briefly discuss the conditions
on its integerness, its stability, or change under the coarse graining
and scaling operations. 
\end{enumerate}
The whole process prehaps resembles the idea of the \textbf{renormalization
group} in the Wilsonian sense. This is of course not accidental. In
\citep{RG} we called a certain coarse graining scheme the \textbf{geometric
renormalization group}. We will complement it with a rescaling process
towards some \textbf{continuum space} developed in \citep{Requ2}.

Before continuing, we should mention that some of the important concepts
are written in bold, and for self-containment, many of them are either
defined in the main body of the text, or are presented in the appendix.

\section{Macro States, Micro States, and Typical States\label{sec:typical}}

It is common practice in the physics of large and complicated systems,
consisting of a huge number of interacting microscopic DoF, to work
with ensembles of states or configurations. The same is the case in
the various approaches to quantum gravity. Frequently something like
a canonical ensemble over micro states is used, the statistical weight
being given by a (pseudo) Hamiltonian, the main purpose of which is
usually defining the exponent in the Boltzmann weight. One should
note, however, that it is not completely clear if the concept of a
Hamiltonian, apart from its mere probabilistic role, is an adequate
notion in fundamental space-time physics.

On the other hand, we observe that macroscopic space-time, as the
stage on which all the usual physical processes are going on, is not
some ensemble in our various model theories. In the following we want
to briefly explain how we see its role in our paper.

In \citep{Observables} we indicated that S-T may rather be called
an \textbf{order parameter manifold} with the non-vanishing metric
tensor $g(x)$ being an \textbf{order parameter field}. The concept
of an order parameter stems from the statistical mechanics of phase
transitions and we think, a similar phenomenon can be seen in our
space-time context if we view S-T as described in the introduction.
We plan to give more details elsewhere, for the time being we content
ourselves to invoke the following picture:
\begin{conjecture}
Classical space-time, S-T, is the macroscopic, long-distance, low-energy
description of an ensemble of microscopic configurations of only incompletely
known (quantum) nature. As in statistical mechanics, we regard this
ensemble of underlying micro states as a phase cell of states which
look alike macroscopically.
\end{conjecture}
To our knowledge, such a dual structure was for the first time analyzed
by v. Neumann in a beautiful paper \citep{Neumann} in the context
of the quantum theory of many DoF. It was recently translated into
English and commented upon in \citep{Lebowitz1}. A particular role
is played by the emergence of commuting macro observables, a topic
which was further developed in \citep{Decoherence}. We think, the
situation in (quantum) space-time physics is very much the same ( Recently this kind of \emph{typicality} was also studied in our context in e.g. 
\citep{Gurau1609} and \citep{Anza1605}).

In a next step we want to argue why we nevertheless are allowed to
deal in our investigation with single states, thus avoiding the intricate
study of ensembles of micro states describing S-T on microscopic scales,
and in particular, their behavior under coarse graining. This point
was also discussed for the first time in \citep{Neumann}, for a recent
discussion see e.g. \citep{Lebowitz2}. Without giving proofs (a more
systematic study can be found in \citep{Ledoux}, some examples are
also discussed in the famous section 3.5 of \citep{Gromov}) we would
like to state the following:
\begin{ob}
It is an important observation that many systems consisting of a great
number of microscopic DoF display the following property: with (very)
high probability their micro states are concentrated in an unusually
small region of phase or configuration space. This allows us to speak
of typical states and employ this concept instead of dealing with
the full ensemble of micro states in a phase cell.
\end{ob}
The basis of these results are the Levy-like inequalities and the
Levy concentration theorems. A typical example is the \textbf{law
of large numbers} and in our context the \textbf{random graph model}
as it was used in \citep{Requlumps}. The crucial point in all these
examples is not that mean value and variance of certain random functions
over some measure space do exist but that the variance is unusually
small. This implies that a random function can be replaced probabilistically
by its average, and configurations by the particular configurations
belonging to the respective mean values.

\section{Graphs as metric spaces\label{sec:Graphs-mspace}}

In this section we gather a few definitions and properties of graphs
for completeness and show that they are metric spaces, so that all
the results regarding metric spaces in this paper applies to them.

A \textbf{graph}, $G(V,E)$, consists of a set of points $V=\left\{ v_{i}\right\} $
that are its \textbf{vertices} and a set of \textbf{edges} $E=\left\{ e_{ij}\right\} \subset V\times V$
that connects some or all of the vertices. If $e_{ij}=e_{ji}$ the
graph is called \textbf{undirected} otherwise it is \textbf{directed}.
The number of edges incident on a vertex $v_{i}$ is called the \textbf{degree}
or \textbf{valency} of that edge, $\deg\left(v_{i}\right)$. If each
vertex of the graph has the same degree $k$ the graph is called a
\textbf{$k$-regular} graph and the graph itself is said to have \textbf{vertex
degree $k$}. If $\deg\left(v_{i}\right)\leq A$ for all $v_{i}\in V(G)$,
then the graph has \textbf{globally bounded} vertex degree. The degree
is called \textbf{locally bounded} if $\forall v_{i}\in V(G)\Rightarrow\deg\left(v_{i}\right)\in\mathbb{N}_{0}$.
A graph is \textbf{connected} if every pair of vertices are connected
by a finite edge sequence.

A \textbf{path} $\gamma$ is an edge sequence without repetition of
vertices with the possible exception of initial and final vertex.
The \textbf{length} $l(\gamma)$ of a path $\gamma$ is the number
of edges occurring in the path. A \textbf{geodesic path} between two
vertices is a path of minimum length. 

One can define a natural \textbf{metric on a graph} by defining the
distance between two vertices $v_{i},v_{j}$ as the length of the
geodesic path between them
\begin{equation}
d_{G}(v_{i},v_{j}):=\min_{\gamma}\{l(\gamma),\gamma\;\text{between}\;v_{i}\;\text{and}\;v_{j}\}.
\end{equation}
Then the graph $G$ together with this distance function is a metric
space $(G,d_{G})$.

Using this one can define a \textbf{closed ball} of radius $r$ centered
around $v_{i}$ in a graph as
\begin{equation}
B(v_{i},r)=\left\{ v_{j}|d\left(v_{i},v_{j}\right)\leq r\right\} ,
\end{equation}
just like the definition in any metric space. Also the \textbf{boundary
of a ball} is defined as
\begin{equation}
\partial B(v_{i},r)=\left\{ v_{j}|d\left(v_{i},v_{j}\right)=r\right\} .
\end{equation}
An \textbf{open ball} is defined as $B(v_{i},r)-\partial B(v_{i},r)$.
One then defines the number of nodes in a ball and in the boundary
of a ball as $|B(v_{i},r)|\,,\,|\partial B(v_{i},r)|$ respectively.

The \textbf{growth function} $\beta(G,v_{i},r)$ starting from a vertex
$v_{i}$ as a function of distance $r$ from it, is defined as
\begin{equation}
\beta(G,v_{i},r)=|B(v_{i},r)|.
\end{equation}
Correspondingly we define 
\begin{equation}
\partial\beta(G,v_{i},r)=\beta(G,v_{i},r)-\beta(G,v_{i},r-1)
\end{equation}
which is the difference between the number of nodes in a ball of radius
$r$ and another one of radius $r-1$ both centered at $v_{i}$. 

The growth function of a graph can have different forms, however,
an important case for our purposes is when
\begin{equation}
\beta(G,v_{i},r)\lesssim r^{\alpha},\,\,\,\,\,\,\,\,\alpha\geq0.
\end{equation}
In that case the graph $G$ is said to have \textbf{polynomial growth}.
For such a graph one can define the degree of polynomial growth as
\begin{equation}
\bar{D}(G,v_{i})=\limsup_{r\to\infty}\frac{\log\beta(G,v_{i},r)}{\log r}.
\end{equation}
If $G$ has locally bounded vertex degree, and if $\exists A,B,\alpha$
such that
\begin{equation}
Ar^{\alpha}\leq\beta(G,v_{i},r)\leq Br^{\alpha},
\end{equation}
for all $r>r_{0}$ for some $r_{0}$ and for all $A,B$ independent
of the reference point $v_{i}$, then we say $G$ has \textbf{uniform
polynomial growth}. The polynomial growth and its uniformness plays
an important role both in convergence of the coarse graining process
and in obtaining an integer dimension for spaces under considerations,
as we will see later.
\begin{bem}
For more details see e.g. \citep{Requ2} and the literature cited
there. We developed such notions already in \citep{Requ3} on physical
grounds without being aware that a similar concept was used in \textbf{geometric
group theory} (\citep{Harpe}).
\end{bem}
\begin{lemma}
For locally bounded vertex degree the exponent $\bar{D}$ is independent
of the vertex $v_{i}$, i.e. is a graph characteristic (see \citep{Requ3}).
\end{lemma}

\section{A generic coarse graining scheme\label{sec:coarse}}

As mentioned in the previous sections, we are going to work in a sufficiently
large class of metric spaces which we will call $\mathcal{S}$. We
would like to introduce a \textbf{coarse graining process} $\mathcal{K}$
in $\mathcal{S}$ such that 
\begin{eqnarray}
\mathcal{K}: & \mathcal{S}\rightarrow & \mathcal{S}\\
 & \left(M_{1},d_{1}\right)\mapsto & \left(M_{2},d_{2}\right),\,\,\,\,\,\,\,\,\left(M_{1},d_{1}\right),\left(M_{2},d_{2}\right)\in\mathcal{S}.
\end{eqnarray}
As we will explain below, we choose $\mathcal{K}$ to be a process
that falls under the category of quasi-isometries. To understand what
this means, we should start by introducing some preliminary concepts.
The first concept is the notion of an isometric embedding. For brevity,
we sometimes write a metric space $\left(M,d\right)$ as just $M$
when the context is clear.

Given two metric spaces $X,Y\in\mathcal{S}$, an \textbf{isometric
embedding} $f$ of $X$ into $Y$ 
\begin{equation}
f:X\rightarrow Y,
\end{equation}
is a distance preserving map, i.e. 
\begin{equation}
d_{X}(x_{1},x_{2})=d_{Y}(f(x_{1}),f(x_{2})),\,\,\,\,\,\forall x_{1},x_{2}\in X.
\end{equation}
It is an \textbf{isometry} if it is also surjective. A weaker but
in our context much more useful and appropriate version is a quasi-isometric
embedding. A map $f:X\rightarrow Y$ between two metric spaces is
called a \textbf{quasi-isometric embedding} if $\exists\lambda\geq1,\epsilon\geq0$
such that 
\begin{equation}
\frac{1}{\lambda}d_{X}(x_{1},x_{2})-\epsilon\leq d_{Y}(f(x_{1}),f(x_{2}))\leq\lambda d_{X}(x_{1},x_{2})+\epsilon,\,\,\,\,\,\,\,\,\,\,\,\,\forall x_{1},x_{2}\in X\label{eq:def-quas-iso-embd}
\end{equation}
i.e. $\exists\lambda\geq1,\epsilon\geq0$ such that $\forall x_{1},x_{2}\in X$,
the distance between their images under $f$, is within a factor $\lambda$
and up to an additive constant of their original distances.

A quasi-isometric embedding $f:X\rightarrow Y$ is called a \textbf{quasi-isometry}
if every point $y\in Y$ lies within a constant distance $C\geq0$
of an image point, 
\begin{equation}
\forall y\in Y:\,\,\exists x\in X:\,\,d_{Y}(y,f(x))\leq C.\label{eq:def-quas-iso}
\end{equation}
Quasi-isometry allows us to compare metric spaces neglecting their
small-scale structure, i.e. it is an \emph{equivalence relation} on
metric spaces ignoring small-scale structure and just looking at the
coarse structures (see example below). Note that here the important
property is the \textbf{transitivity}, i.e., if $X,Y$ and $Y,Z$
are quasi-isometric then $X,Z$ are also quasi-isometric. In this
context one should perhaps mention the approach of \citep{Menger}
on \textbf{random metric spaces}. If in the above definition, $\lambda=1$,
then the map $f$ is called a \textbf{rough isometry}.
\begin{bem}
In the following we will usually choose for convenience $C=\epsilon$
and call it an $\epsilon$-\textbf{rough isometry}. This is useful
as we frequently study cases where both $\epsilon,C$ approach zero.
\end{bem}
Quasi-isometry can actually be restated in a symmetric way. Two metric
spaces $X$ and $Y$ are quasi-isometric if there exists $\lambda\geq1,\epsilon\geq0,\rho\geq0$
and maps 
\begin{align}
f:X\rightarrow Y, &  & g:Y\rightarrow X
\end{align}
such that for all $x,x_{1},x_{2}\in X$ and $y,y_{1},y_{2}\in Y$,
\begin{align}
d_{Y}\left(f(x_{1}),f(x_{2})\right)\leq\lambda d_{X}\left(x_{1},x_{2}\right)+\epsilon, &  & d_{X}\left(g(y_{1}),g(y_{2})\right)\leq\lambda d_{Y}\left(y_{1},y_{2}\right)+\epsilon
\end{align}
and 
\begin{align}
d_{X}\left(g\circ f(x),x\right)\leq\rho, &  & d_{Y}\left(f\circ g(y),y\right)\leq\rho.
\end{align}
Such a map $g$ is called a \textbf{quasi-inverse} (see e.g. \citep{Requ2}
or \citep{Harpe,Bridson})

As a simple but interesting example of quasi-isometry, we note that
the integer lattice $\mathbb{Z}^{n}$ is quasi-isometric to $\mathbb{R}^{n}$.
This can be shown by considering the quasi-isometry map 
\begin{align*}
f: & \mathbb{Z}^{n}\rightarrow\mathbb{R}^{n}\\
: & (x_{1},\ldots,x_{n})\mapsto(x_{1},\ldots,x_{n}),\,\,\,\,\,\,\,\,\,\,x_{i}\in\mathbb{Z}.
\end{align*}
The map is metric-preserving if we use the Euclidean metric in both
spaces. It is perhaps more natural if we employ the intrinsic \textbf{graph
metric} (defined below) on $\mathbb{Z}^{n}$, that is, the minimal
number of steps on the lattice $\mathbb{Z}^{n}$. In that case we
have a true quasi-isometry. In the former case we have $\lambda=1$
and $\epsilon=0$, while $n$-tuples $\in\mathbb{R}^{n}$ are within
a distance $C=\sqrt{\frac{n}{4}}$ of $n$-tuples $\in\mathbb{Z}^{n}$
with respect to the usual Euclidean metric. On the other hand, the
map 
\begin{align*}
g: & \mathbb{R}^{n}\rightarrow\mathbb{Z}^{n}\\
: & (y_{1},\ldots,y_{n})\mapsto(x_{1},\ldots,x_{n}),\,\,\,\,\,\,\,\,\,\,y_{i}\in\mathbb{R}
\end{align*}
rounding $n$-tuples $\in\mathbb{R}^{n}$ to the nearest $n$-tuple
$\in\mathbb{Z}^{n}$ is also a quasi-isometry. In this case, the distance
between pairs of points is changed by adding or subtracting at most
$2\sqrt{\frac{n}{4}}$. Also we can see that 
\begin{equation}
d_{X}\left(g\circ f(x),x\right)=0
\end{equation}
and 
\begin{equation}
d_{Y}\left(f\circ g(y),y\right)\leq\sqrt{\frac{n}{4}}
\end{equation}
which means that for this example, $\rho=\sqrt{\frac{n}{4}}$.

Now we are in a position to make our definition of coarse graining
more precise. In our approach, a \textbf{coarse graining map $\mathcal{K}$
is a quasi-isometry} $\mathcal{K}:X\rightarrow Y$, while not all
quasi-isometries can be regarded as coarse grainings. The particular
form of such a coarse graining has of course to be motivated by the
physical context. It implies that details or finer DoF are deleted
or summed over and that we go over to coarser substructures in each
step. Therefore not every quasi-isometry will do.

Using a coarse graining operation $\mathcal{K}$, we can generate
a sequence of metric spaces, $\left\{ \left(M_{i},d_{i}\right)\right\} ,i=0,\ldots,N$,
such that,
\begin{equation}
\left(M_{0},d_{0}\right)\stackrel{\mathcal{K}}{\longrightarrow}\left(M_{1},d_{1}\right)\stackrel{\mathcal{K}}{\longrightarrow}\left(M_{2},d_{2}\right)\stackrel{\mathcal{K}}{\longrightarrow}\cdots\stackrel{\mathcal{K}}{\longrightarrow}\left(M_{N},d_{N}\right),
\end{equation}
where all $\left(M_{i},d_{i}\right)\in\mathcal{S}$, and $\left(M_{i+1},d_{i+1}\right)$
is a coarse grained version of $\left(M_{i},d_{i}\right)$ in this
sequence. We also write
\begin{equation}
\left(M_{i+j},d_{i+j}\right)=\underbrace{\mathcal{K}(\mathcal{K}(\cdots\mathcal{K}}_{j\,\textrm{times}}((M_{i},d_{i}))))=\mathcal{K}^{j}\left(\left(M_{i},d_{i}\right)\right).
\end{equation}
Then the main idea is that after sufficiently many steps of such coarse
grainings, one will hopefully arrive at a final limit space that is
stable under further coarse graining $\mathcal{K}$, similar to a
fixed point in the renormalization group analysis\footnote{Such a program was motivated in e.g. \citep{RG}.}.
Thus we expect this ``\textbf{fixed point}'' (see below), if it
does exist, to exhibit certain \textbf{selfsimilarity} properties
as was studied in \citep{scalefree}. Furthermore, as in some of the
renormalization group construction on, say, a lattice, all the intermediate
spaces remain discrete, and it is only after a rescaling procedure
that the fixed point is expected to behave as a smooth structure,
resembling our macroscopic space-time.

One should say that a fixed point in the strict sense may not exist.
We rather expect that we have what we call a \textbf{macroscopic fixed
point} which comprises an ensemble of \textbf{microstates} looking
the same macroscopically. This macroscopic fixed point is expected
to be stable under further coarse graining transformations while the
microstates move in the class belonging to the macroscopic fixed point.

\subsection{Some concrete examples of coarse graining $\mathcal{K}$}

It is illustrative and important to describe various graph transformations
which are physically relevant to coarse graining and belong to the
class of quasi-isometries or rough isometries. Two of such operations,
\textbf{edge insertions} and \textbf{edge deletions}, we presented
in observation 2.7 of \citep{Requ2}. There we employed these processes
in the context of a notion of dimension of networks/graphs. One can
easily see that the respective proofs, which can be found in \citep{RG}
and \citep{Requ3}, apply also in our case of quasi-isometry \footnote{At the time of writing \citep{RG,Requ3} we were not aware of the
various mathematical notions we used in \citep{Requ2}.}. Below we will present a few more of such transformations that constitute
coarse graining in the sense that they represent a transition $G\to G^{\prime}$
in which some of the finer degrees of freedom in $G$ are ignored
or averaged over. 

\subsubsection{Quasi-isometry coarse graining}

\subsubsection*{$k$-local edge insertion/deletion}

A \textbf{$k$-local insertion/deletion of edges} is the insertion/deletion
of arbitrarily many edges in the $k$-neighborhood of vertices of
a locally finite graph $G$. In the case of edge deletions the procedure is slightly
more involved, that is, it refers rather to $k$-neighborhoods in
the new graph $G'$. I.e., edges are deleted between vertices which have a distance smaller than $k$ in $G'$. It is tacitly assumed that $G'$ is still connected. Note furthermore that the maximal number of possible edge insertions in a $k$-neighborhood of a vertex is bounded in a locally finite graph.
The resulting graph $G^{\prime}$ is quasi-isometric
to $G$ (for more details see Observation 2.7 in \citep{Requ2}).

\subsubsection*{Vertex contraction of diameter $\leq k$}

Another graph transformation $G\rightarrow G^{\prime}$ consists of
the following steps. In $G$, take the subgraphs $H_{i}^{G_{k}}$of
\textbf{diameter} $\leq k$ and contract them to a single vertex.
These will be the vertices $v_{i}^{\prime}\in G^{\prime}$. An edge
$e_{ij}^{\prime}\in E(G^{\prime})$, pointing from $v_{i}^{\prime}\in G^{\prime}$
to $v_{j}^{\prime}\in G^{\prime}$ is drawn in $G^{\prime}$, if there
are edges in $G$ that point from $H_{i}^{G_{k}}$ to $H_{j}^{G_{k}}$.
In case several edges in $G$ point from $H_{i}^{G_{k}}$ to $H_{j}^{G_{k}}$,
we have two options. Either they are replaced in $G'$ by a single
edge, or they can inhere an edge color based on the number (and/or
color) of edges in $G$ pointing from $H_{i}^{G_{k}}$ to $H_{j}^{G_{k}}$,
that are identified as the edge $e_{ij}^{\prime}\in G^{\prime}$.
The latter is an interesting options which makes it possible that
the coarse grained graph $G^{\prime}$ can have different (and coarse
grained) edge colors that are inherited from the underlying graph
$G$. This may have interesting and possibly deep consequences in
theories like loop quantum gravity and group field theory where the
color of edges represent the irreducible representation of the local
gauge group of the macroscopic theory.
\begin{conclusion}
All these $k$-local graph transformations lead to a quasi-isometric
new graph $G'$ and in general they are true quasi-isometries (i.e.,
not rough isometries), in particular if they are performed globally
on the infinite graph $G$.
\end{conclusion}

\subsubsection{Rough isometry coarse graining}

\subsubsection*{Clique graph transformation $\mathcal{C}(G)$}

This is the transition from a graph $G$ to its so called \textbf{clique
graph} $\mathcal{C}(G)$. A \textbf{Cliques} or a maximals \textbf{subsimplex}
in a graph $G$ is a subset of $G$ in which all the vertices are
connected with each other. Cliques and clique graphs play an important
role in graph theory (cf. e.g. \citep{Bollo} or \citep{Pisana1,Pisana2})
and we used them already in \citep{RG,Requlumps}. Our motivation
was mainly physical as we tried to find a substitute for the block
spins of the ordinary Wilsonian renormalization group in strongly
erratic and disordered systems.

The \textbf{clique graph} $\mathcal{C}(G)$ of $G$ is defined as
a graph whose vertices are the cliques of $G$, and an edge exists
between two of vertices of $\mathcal{C}(G)$ if the underlying cliques
have non-zero vertex overlap (in $G$). Notice that in this case,
the number of these vertices in $G$ that are shared between two cliques,
can give rise to colors of edges in $\mathcal{C}(G)$ in different
manners, hence leading to emergent edge colors in $\mathcal{C}(G)$.
However, in this case, these emergent colors do not seem to have anything
to do with edge colors in $G$, but rather the information related
to the vertices of $G$. Now we present an important theorem.
\begin{theo}\label{Theo:Spencer}
If $G$ has globally bounded vertex degree $deg(G)$, its clique graph
$\mathcal{C}(G)$ has also globally bounded vertex degree, and is
roughly isometric to $G$ with $C=\epsilon=1$.
\end{theo}
See Appendix \ref{AppSec:Proofs} for a proof.

\section{The Gromov-Hausdorff Space of Metric Spaces\label{sec:supermetric space}}

Having defined a coarse graining procedure $\mathcal{K}$ which generates
a coarse grained sequence $\left\{ \left(M_{i},d_{i}\right)\right\} _{i=0,\ldots,N}^{\mathcal{K}}$,
it is possible to introduce some kind of a distance function $d_{\mathcal{S}}$
between suitable metric spaces, with respect to which one can ask
if the sequence converges to a limit space, and to what extent two
spaces are structurally similar or different.

We begin with the simpler notion of the distance of two subsets of
a metric space. Given a metric space $\left(Z,d_{Z}\right)$ and two
of its nonempty subsets $X,Y\subset Z$, one can define a distance
between these subsets called the \textbf{Hausdorff distance} $d_{H}^{Z}$
as
\[
d_{H}^{Z}\left(X,Y\right)=\max\left\{ \sup_{x\in X}\,\inf_{y\in Y}\,d_{Z}(x,y),\sup_{y\in Y}\,\inf_{x\in X}\,d_{Z}(x,y)\right\} 
\]
or
\[
d_{H}^{Z}\left(X,Y\right)=\inf\left\{ \epsilon\geq0|X\subseteq U_{\epsilon}(Y),Y\subseteq U_{\epsilon}(X)\right\} .
\]
Here, $U_{\epsilon}(X)$ is the \textbf{$\epsilon$-neighborhood of
a subset} $X\subset Z$ of a metric space $(Z,d_{Z})$, and it is
defined as 
\[
U_{\epsilon}(X)=\bigcup_{x\in X}\left\{ z\in Z|d(z,x)\leq\epsilon\right\} ,
\]
i.e., it is the union of all $\epsilon$-balls around all $x\in X$,
or the set of all points in $Z$ that are within a distance $\epsilon$
of the set $X$, or the generalized ball of radius $\epsilon$ around
$X$. 

The Hausdorff distance makes the set of non-empty \textbf{compact}
subsets of a \textbf{complete} metric space into a complete metric
space \citep{Edgar,Bridson}. Note that on the set of all non-empty
(not necessarily compact) subsets of $Z$, in general $d_{H}^{Z}$
only defines a \textbf{pseudometric}, viz., with $A,B\subset Z$,
then $d_{H}^{Z}(A,B)=0$ does not necessarily mean $A\neq B$.

With the notions of the Hausdorff distance and isometric embedding,
Gromov \citep{Gromov,Bridson,Petersen} was able to develop a distance
concept between two arbitrary compact metric spaces. This distance
is called the \textbf{Gromov-Hausdorff distance $d_{GH}$} between
two compact metric spaces $X,Y$ and is defined as
\begin{equation}
d_{GH}\left(X,Y\right)=\inf\,d_{H}^{Z}\left(f(X),g(Y)\right)\label{eq:GH-dist-1st-def}
\end{equation}
for all metric spaces $Z$ and all isometric embeddings $f:X\rightarrow Z$
and $g:Y\rightarrow Z$. Equivalently it is defined as
\begin{equation}
d_{GH}\left(X,Y\right)=\inf\,d_{H}^{X\sqcup Y}\left(f(X),g(Y)\right),\label{eq:GH-dist-2nd-def}
\end{equation}
where $X\sqcup Y$ is the disjoint union of $X,Y$, and the metric
$d_{H}^{X\sqcup Y}$ extending the metrics on $f(X),g(Y)$. This latter
version was particularly used in \citep{Petersen}. 

This distance has some interesting properties:
\begin{enumerate}
\item On the set of all isometry classes of compact metric spaces, it provides
a metric\footnote{cf. lemma 4.6 in \citep{Requ2}.}. This set together
with $d_{GH}$ is a complete metric space called the \textbf{Gromov-Hausdorff
space} \citep{Petersen}. 
\item In general, however, it provides only a pseudo-metric. This holds
for example if $X$ is dense in $Y$, or if $X,Y$ are isometric.
\item It measures how far two compact metric spaces are from being isometric,
i.e. $X,Y$ are isometric iff $d_{GH}(X,Y)=0$. It is is in fact a
measure of metric similarity.\label{enu:dGH-iso-distance}
\item It defines a notion of convergence for sequences of compact metric
spaces, called the \textbf{Gromov-Hausdorff convergence}.
\end{enumerate}
In the Gromov-Hausdorff space, a metric space to which a sequence
of compact metric spaces converges in $d_{GH}$ is called its \textbf{Gromov-Hausdorff
limit}\footnote{For brevity we may use GH instead of Gromov-Hausdorff in some parts
of the text.}.

We now see that given these notions, one can (in principle) check
if in the Gromov-Hausdorff space, a sequence of metric spaces, generated
by a concept of coarse graining $\mathcal{K}$, or \textbf{scale transformation}
(see below), has a Gromov-Hausdorff limit. 

However, many of the spaces that are interesting physically are non-compact
and hence do not belong to the Gromov-Hausdorff space. Thus we need
to somehow extend these notions to sequences of coarse grained spaces,
$\left\{ \left(M_{i},d_{i}\right)\right\} _{i=0,\ldots,N}^{\mathcal{K}}$,
that are non-compact, at least to a relevant subset of these non-compact
metric spaces. This was also done by Gromov. More precisely, he extended
this notion of convergence in $d_{GH}$, to non-compact but locally
compact metric spaces $\mathcal{S}$. To see this, we need the concept
of pointed metric spaces:
\begin{defi}
A \textbf{locally compact complete} metric space $\left(X_{i},d_{i}\right)$
with a distinguished point $x_{i}\in X_{i}$, is called a pointed
metric space. It is denoted by $\left(\left(X_{i},d_{i}\right),x_{i}\in X_{i}\right)$.
\end{defi}
The extension of the Gromov-Hausdorff convergence to $\mathcal{S}$
is then done in the following way: 
\begin{defi}
A sequence of pointed metric spaces $\left\{ \left(\left(X_{i},d_{i}\right),x_{i}\in X_{i}\right)\right\} $
is defined to converge to $\left(\left(X,d_{X}\right),x\in X\right)$,
if for all $r>0$, the sequence of closed $r$-balls, $\left\{ B(x_{i},r)\subset X_{i}\right\} $,
converges to $B(x,r)\subset X$ in $d_{GH}$. This is called \textbf{pointed
GH-convergence}.
\end{defi}
Thus in effect, if we have a sequence of non-compact but locally compact
metric spaces, we are still able to draw a conclusion if they converge
to a locally compact space in pointed GH-sense. We take this space
of non-compact but locally compact metric spaces as our super space
$\mathcal{S}$ with which we will work. It is quite similar in concept, to the ``theory space'' of the renormalization methods. As we will see in section
\ref{sec:Converg-2}, this extension of Gromov-Hausdorff convergence
to this superspace  plays a crucial role in our enterprise. It is
worth mentioning that two very interesting types of non-compact spaces
that are locally compact are $\mathbb{R}^{n}$, and consequently finite
dimensional non-compact topological manifolds, since they share the
local properties of the Euclidean spaces\footnote{$\mathbb{R}^{n}$ is locally compact due to the Heine-Borel theorem}.

The concepts of GH-convergence or GH-distance seem to be quite abstract
compared to the simpler concept of the Hausdorff-distance, but note
that it yields much more detailed information about the structure
of the spaces under discussion. Hausdorff-distance simply measures
the metrical distance of sets as subsets of a larger metric space.
However, the GH-distance, due to the incorporation of \emph{all admissible
Hausdorff distances}, actually measures the structural relatedness
(or similarity) of spaces. This is much more specific, and in line
with the general notion of coarse graining in physics.

Due to computational complications, it is seldom possible to calculate
the exact GH-distance between two metric spaces\footnote{This is a consequence of the need to incorporating all admissible
metrics (Hausdorff distances), which ironically makes the $d_{GH}$
notion so powerful.}. However, it is usually possible and sufficient to obtain efficient
upper bounds. In this respect the second version of GH-distance, equation
(\ref{eq:GH-dist-2nd-def}), is quite useful\footnote{To get more acquainted with the technical subtleties see section 4
of \citep{Requ2}. The crucial point is always the verification of
the \emph{triangle inequality}.}. 

To get a better feeling how quasi-isometry or rough isometry and GH-distance
are structurally related we present an important theorem\footnote{This is the theorem 4.15 in \citep{Requ2}. It is also contained in
the form given here in \citep{Lochmann}} that has far reaching consequences and plays an important role in
discussions about the convergence.
\begin{theo}
Two metric spaces X,Y, have finite GH-distance iff they are roughly
isometric (i.e., with the $\lambda=1$ and $C=\epsilon$ in (\ref{eq:def-quas-iso-embd})
and (\ref{eq:def-quas-iso})). Furthermore it holds in particular:
\begin{equation}
\frac{1}{2}\,d_{GH}(X,Y)\leq\inf\{\epsilon\}\leq2\,d_{GH}(X,Y)
\end{equation}
with the $\epsilon$'s belonging to $\epsilon$-rough isometries between
$X,Y$.\label{rough}
\end{theo}
The proof along with a brief discussion can be found in \citep{Requ2}. 

An implication of this theorem is that since a finite GH-distance
between $X,Y$, is equivalent to $X,Y$ being $\epsilon$-roughly
isometric, 
\begin{equation}
|d_{Y}(f(x),f(x'))-d_{X}(x,x')|\leq\epsilon,
\end{equation}
then quasi-isometric spaces, when the quasi-isometry is not a rough
isometry, have an infinite GH-distance. That means if $\mathcal{K}$
is such an operation performed on space $G$, then $d_{GH}(G,\mathcal{K}(G))=\infty$
and thus $G$ and $\mathcal{K}(G)$ are structurally quite different
metric-wise. They are in fact infinitely apart from being isometric
(see item \ref{enu:dGH-iso-distance} above). Another observation
is that, a sequence
\begin{equation}
\left\{ G,\mathcal{K}(G),\mathcal{K}\left(\mathcal{K}(G)\right),\ldots\right\} 
\end{equation}
for which all the steps of coarse graining are true quasi-isometries
cannot converge in $d_{GH}$. We will discuss this point in more detail
in the following sections. Note that a point in the coarse graining
sequence where true quasi-isometry changes to rough isometry may be
called a geometric phase transition point as we recognize a transition
from structural dissimilar to structural similar spaces.

Another illuminating observation is the following. Assume that we
have a sequence of metric spaces, $\left\{ X_{\nu}\right\} $, which
converges in $d_{GH}$ towards some metric space $X$. Then by the
same token, there exists a sequence of rough $\epsilon_{\nu}$-isometries,
$f_{\nu}$, between $X_{\nu}$ and $X$ with $\epsilon_{\nu}\to0$.
This implies (by definition) that
\begin{equation}
\lim_{\epsilon_{\nu}\to0}f_{\nu}(X_{\nu})\;\text{is dense in}\;X
\end{equation}
with
\begin{equation}
\lim_{\epsilon_{\nu}\to0}\;d_{X}(f_{\nu}(x_{\nu}),f_{\nu}(x'_{\nu}))=\lim_{\epsilon_{\nu}\to0}\;d_{X_{\nu}}(x_{\nu},x'_{\nu}).
\end{equation}
Thus we say that in the limit $\epsilon_{\nu}\to0$, the spaces $X_{\nu},X$
become \textbf{essentially isometric}, in order to have a label for
this asymptotic behavior.

\section{Convergence and Continuum Limit I\label{sec:Converg-1}}

The deep question of under what physical and/or mathematical conditions
a sequence of coarse-grained spaces has a macroscopic or continuum
limit, will be postponed to the next section. In this section we will
discuss various topics related to the rescaling of our metric spaces,
and the fixed point in this rescaling process (some results can also
be found in \citep{Requ2,Lochmann}).

Since we will be using some concepts from dynamical systems, let us
begin by some related definitions. In a dynamical system with a space
of admissible states (or phase space) $H$, an \textbf{attractor}
or an \textbf{attracting set} is, roughly speaking, a closed subset
$A\subset H$ such that for many choices of initial states $I_{A}\subset H$,
the system will eventually evolve to $A$. The set of initial conditions
$I_{A}$ for which the system's state eventually evolves to $A$ is
called the \textbf{basin of attraction} of $A$. 

Now we explore the relation of the above notions with our framework.
If we have a metric space, $X$, and a metric $d_{X}$ on it, we can
define, in a canonical way, a whole sequence of \textbf{scaled} metrics
$\lambda\cdot d_{X}$ with $\lambda\in\mathbb{R}^{+}$. The limit
$\lambda\to0$ corresponds to the large scale structure of $X$, while
$\lambda\to\infty$ reveals the fine structure of $X$ by magnifying
the infinitesimal neighborhoods of the points of $X$. In our context
the limit $\lambda\to0$ is of particular importance.

We assume that the GH-limit
\begin{equation}
\left(X_{\infty},d_{\infty}\right):=\lim_{\lambda\to0}(X,\lambda d_{X})
\end{equation}
exists with metric $d_{\infty}=\lim_{\lambda\to0}\lambda d_{X}$,
and we want to infer some general properties of this limit space.
\begin{ob}
From what we have learned in the last section, all spaces, \\
$\left\{ \left(X^{\prime},d_{X^{\prime}}\right)|d_{GH}\left(X^{\prime},X\right)<\infty\right\} $
have the same limit $\left(X_{\infty},d_{\infty}\right)$. Furthermore,
it is easy to see that $\left(X_{\infty},d_{\infty}\right)$ is the
only scaling limit in this set.
\end{ob}
This set of spaces, $\left\{ \left(X^{\prime},d_{X^{\prime}}\right)\right\} $
(including $\left(X,d_{X}\right)$ itself), is the basin of attraction
of the attractor $\left(X_{\infty},d_{\infty}\right)$ under the evolution
map
\begin{equation}
\phi_{\lambda}:\;(X',d_{X'})\,\longmapsto(X',\lambda d_{X'})
\end{equation}
for $\lambda\to0$. This limit space $\left(X_{\infty},d_{\infty}\right)$
has the following nice property:
\begin{ob}
$\left(X_{\infty},d_{\infty}\right)$ is \textbf{scale invariant}
under every scaling map 
\begin{equation}
\phi_{l}:\;(X_{\infty},d_{\infty})\longmapsto(X_{\infty},ld_{\infty})
\end{equation}
in the sense that 
\begin{equation}
d_{GH}(X_{\infty},lX_{\infty})=0.
\end{equation}
This implies that there exists an essentially isometric map for every
$\phi_{l}$ and, as a consequence, a scaling map, $f_{l}$, from $X_{\infty}\to X_{\infty}$,
i.e., we have
\begin{equation}
d_{\infty}(x,x')=l\cdot d_{\infty}(f_{l}(x),f_{l}(x')).
\end{equation}
\end{ob}
\begin{proof}
With $\lim_{\lambda\to0}(X,\lambda d_{X})=\left(X_{\infty},d_{\infty}\right)$,
it holds that $\lim_{l\lambda\to0}(X,l\lambda d_{X})=\left(X_{\infty},d_{\infty}\right)$
in GH-sense. On the other hand, we have $\lim_{l\lambda\to0}(X,l\lambda d_{X})=l\cdot\lim_{\lambda\to0}(X,\lambda d_{X})=l\cdot\left(X_{\infty},d_{\infty}\right)$.
\end{proof}
We now see the following:
\begin{conclusion}
One may call $\left(X_{\infty},d_{\infty}\right)$ a \textbf{fixed
point} of the scaling map $\phi_{\lambda}$ for $\lambda\to0$ in
its basin of attraction given by $\left\{ \left(X^{\prime},d_{X^{\prime}}\right)|d_{GH}\left(X^{\prime},X\right)<\infty\right\} $.
\end{conclusion}
Some examples of scale invariant spaces are e.g. $\mathbb{R}^{n}$
or various \textbf{fractal spaces}.

\section{Convergence and Continuum Limit II\label{sec:Converg-2}}

In this section we want to develop criteria under which a sequence
of metric spaces has a limit space under $\mathcal{K}$. Crucial in
this respect is the Gromov-compactness theorem. As this argument is
quite intricate and was already discussed in \citep{Requ2} we will
only briefly recapitulate the relevant points for the sake of completeness
and will, refer the reader to \citep{Requ2} for more details.

A family of compact spaces, $X_{\lambda}$, is called \textbf{uniformly
compact} if their diameters are uniformly bounded and if for each
$\epsilon>0$, $X_{\lambda}$ is coverable by $N_{\epsilon}<\infty$
balls of radius $\epsilon$ independent of the index $\lambda$. We
then have the fundamental result, derived by Gromov \citep{Gromov1}:
\begin{theo}
A sequence of metric spaces $\left\{ \left(X_{i},d_{i}\right)\right\} $
contains a convergent subsequence in $d_{GH}$, iff it is uniformly
compact.
\end{theo}
The proofs typically use an \emph{Arzela-Ascoli-Cantor-diagonal-sequence}-like
argument \citep{Gromov1,Petersen,Bridson}. This theorem can immediately
be extended to the sequences of pointed metric spaces in $\mathcal{S}$.
\begin{theo}
If for all $r$ and $\epsilon>0$ the balls $B\left(x_{i},r\right)$
of a given sequence of proper metric spaces $\left\{ (X_{i},x_{i}\in X_{i})\right\} $
are uniformly compact, then a subsequence of spaces converges in pointed
GH-sense.
\end{theo}
Note that this means that the balls converge in the usual GH-sense.
But the convergence is not uniform.

It is perhaps helpful to illustrate these results by giving a simple
example. Take the lattice $\mathbb{Z}^{n}$ embedded in $\mathbb{R}^{n}$
and take the scaling limit 
\begin{equation}
\phi_{l}:\left(\mathbb{Z}^{n},d_{\mathbb{Z}^{n}}\right)\longmapsto\left(\mathbb{Z}^{n},\lambda d_{\mathbb{Z}^{n}}\right),\,\,\,\,\,\,\,\,\,\,\,\,\,\lambda=2^{-l}
\end{equation}
where $d_{\mathbb{Z}^{n}}$ is a suitable metric on $\mathbb{Z}^{n}$
(see below). For $\lambda\rightarrow0$, i.e. $l\to\infty$, we have
\begin{equation}
\lim_{\lambda\to0}\left(\mathbb{Z}^{n},\lambda d_{\mathbb{Z}^{n}}\right)=\left(\mathbb{R}^{n},d_{\mathbb{R}^{n}}\right),
\end{equation}
which holds only in pointed GH-sense since the convergence is not
uniform. Here the metric of the scaling limit, $d_{\mathbb{R}^{n}}$, depends on $d_{\mathbb{Z}^{n}}$. For example if we use the Euclidean metric on $\mathbb{Z}^{n}$,
the limit space $\mathbb{R}^{n}$ also carries the ordinary Euclidean
metric. But if we use the \textbf{graph metric} or \textbf{taxicab
metric} on $\mathbb{Z}^{n}$, the limit metric is also the taxicab
metric (or $l^{1}$-metric) on $\mathbb{R}^{n}$. In this example, for a fixed ball around $x=0$ we can infer from what we
said in theorem \ref{rough} that for $l\to\infty$ the ball is more
and more filled with points stemming from lattices having edge length
$2^{-l}$. It is in this way that we can envisage the pointed GH-convergence
in the scaling situation.

To use this theorem effectively, we need practical and easy to control
properties, that imply that a sequence of spaces is uniformly compact.
We will supply properties which hold in particular in the situation
we are interested in, that is, (infinite) networks/graphs. We begin
with versions of the \emph{doubling property}. A metric space is called
\textbf{doubling} if each ball, $B(x_{i},r)$, can be covered by at
most $C$ balls with radius half that of $B(x_{i},r)$, with $C$
independent of the balls $B$. It easily follows via iteration that
this implies that $B$ is coverable by $C^{k}$ balls of radius $2^{-k}$.

There is a nice relation of this property to a more manageable case
as follows: With $(X,d)$ a metric space, and $\mu$ a positive Borel
measure on $X$, then $\mu$ is said to be \textbf{doubling}, if there
exists a positive constant, $C$, being independent of $B$ such that
\begin{equation}
\mu(2B)\leq C\cdot\mu(B)
\end{equation}
for all balls in $X$. Here $2B$ is a ball with the same center as
$B$ but twice the radius. Then it follows that 
\begin{equation}
\mu(2^{k}B)\leq C^{k}\cdot\mu(B).
\end{equation}
We then have the theorem
\begin{theo}
If $(X,d)$ has a doubling measure, it is doubling as a metric space.
\end{theo}
For a proof see \citep{Gromov}, p.412 (the chapter being written
by Semmes).

Now, defining a sequence of spaces $\left\{ \left(X_{n},d_{n}\right)\right\} $
to be \textbf{uniformly doubling} if the above doubling properties
hold uniformly in it, we have
\begin{theo}
A sequence of spaces is uniformly compact if it is uniformly doubling.
\end{theo}
It turns out that the latter property is more manageable than the
former one: Consider a graph $\left(G,d\right)$ of \textbf{uniform
polynomial growth}, for which
\begin{equation}
Ar^{d}\leq\beta(x,r)\leq Br^{d}
\end{equation}
for all $r\geq r_{0}$ and $A,B$ independent of the reference point
$x$. Taking the sequence of scaled graphs $\left\{ \left(G_{n},d_{n}=n^{-1}\cdot d\right)\right\} $
made from $G$, we can prove (\citep{Requ2})
\begin{conclusion}
A graph with uniform polynomial growth has a doubling \textbf{counting
measure} for sufficiently large $r\geq r_{0}$ and is hence doubling
as a metric space for sufficiently large $r\geq r_{0}$. This implies
that all balls $B_{n}(x,r)$ in $\left\{ G_{n},d_{n}\right\} $, where
$G_{n}$ are of uniform polynomical growth, are uniformly compact.
We conclude that there exists a subsequence of $\left\{ \left(G_{n},d_{n}=n^{-1}\cdot d\right)\right\} $,
that converges in pointed GH-sense.\label{thm:upg-converge}
\end{conclusion}
This conclusion mean that the continuum limit, $\left(G_{\infty},d_{\infty}\right)$,
discussed in the preceding subsection, exists for graphs of uniform
polynomial growth with rescaling map
\begin{equation}
\phi_{\lambda}:(G,d)\longmapsto(G,\lambda d)=(G,n^{-1}d)=(G_{n},d_{n})
\end{equation}
where
\begin{equation}
\left(G_{\infty},d_{\infty}\right)=\lim_{\lambda\to0}\phi_{\lambda}\left((G,d)\right)=\lim_{\lambda\to0}(G,\lambda d).
\end{equation}

\section{The Geometric Renormalization Group in the Superspace of Metric Spaces\label{sec:Geo-Renorm}}

We have now the necessary methods at our disposal in order to develop
a geometric version of a Wilsonian renormalization group (RG) in our
superspace of metric spaces. These consist of a general concept of
\textbf{coarse graining}, a notion of \textbf{continuum limit}, and
the idea of \textbf{typicality} of micro states lying in a \textbf{phase
cell} describing some \textbf{macro state} like our continuum space
time.

To begin with, we have argued in section \ref{sec:typical} that,
instead of dealing with possibly a complicated ensemble structure
in our superspace $\mathcal{S}$, we can perform the coarse graining
process on certain selected micro states, for example networks or
graphs. This coarse graining process was then described in section
\ref{sec:coarse}, at least as far as its general characteristics
are concerned. This process may may differ slightly from one metric
space to another, depending on the type of the space, but the central
pieces are an averaging and/or purification of certain substructures. 

In the cases which interest us most, i.e. networks/graphs, this averaging
consists typically of the replacement of particular subgraphs (like
cliques), by vertices on the next coarse graining level, as described
in section \ref{sec:coarse}. The purification consists of adding/deleting
edges or even whole subgraphs according to certain principles. In
\citep{RG} the substructures were cliques and we deleted cliques
which were unusually small, or edges in the \textbf{clique graph}
if the overlap of the respective cliques was too marginal. The whole
process tries to simulate the \textbf{block spin} approach of the
ordinary real space Wilsonian renormalization group with the cliques
in our example representing the blocks.
\begin{ob}
Note that the repetition of this coarse graining process does not
leave the subregime of discrete graphs/networks. Thus we have to supplement
it by a second type of process, a rescaling, as described in the two
preceding sections. This then yields a continuum limit not only for
the final limit space of a coarse graining process, but also for the
various stages (or spaces) before that.
\end{ob}
Furthermore, in contrast to the ordinary Wilsonian RG, which typically
lives on simple Bravais lattices, the latter process of rescaling
is also quite complicated to perform on highly irregular spaces.

Starting from some initial graph/network $G_{0}$, which we presume
represents the pregeometry of our space-time on the most fundamental
level, and neglecting the possible micro states carried by the vertices
and edges\footnote{As mentioned before, in this first work, we take the simplest cases
where the color of vertices and edges play no role in the coarse graining
scenario.}, we apply a sequence of coarse graining operations $\mathcal{K}$
on $G_{0}$, which then yields a coarse grained sequence of spaces
$\left\{ \left(G_{i},d_{i}\right)\right\} ,\,i=0,\ldots,N$, depicted
as 
\begin{equation}
\left(G_{0},d_{0}\right)\stackrel{\mathcal{K}}{\longrightarrow}\left(G_{1},d_{1}\right)\stackrel{\mathcal{K}}{\longrightarrow}\left(G_{2},d_{2}\right)\stackrel{\mathcal{K}}{\longrightarrow}\cdots\stackrel{\mathcal{K}}{\longrightarrow}\left(G_{n},d_{n}\right).
\end{equation}
According to our assumptions, these coarse graining operators will
belong to the class of \textbf{quasi-isometries}, including rough
isometries. 
\begin{bem}
As we will mention at the end of section \ref{Dim}, we possibly may
have to deal with models of S-T which have both a near- and a far-order
structure on a primordial scale, being generated by a sparse network
of \textbf{translocal} edges. In that case it may happen that we have
to leave the class of quasi-isometric coarse graining which was based
on $k$-local operations.
\end{bem}
We expect that, at least in the first steps, $\mathcal{K}$ will consist
of a large number of subgraph contractions and edge deletions/insertions
according to our fixed coarse graining protocol. Hence, in general,
$\mathcal{K}$ will be a true quasi-isometry and not a rough isometry.
As we saw above, in that case the GH-distance between consecutive
graphs is infinite and thus they will differ structurally. On the
other hand we observed in \citep{RG} that after a number of coarse
graining steps the corresponding graphs had the tendency of becoming
more regular and structurally more similar. This motivates us to formulate
our central RG-conjecture: 
\begin{conjecture}
If we start from suitable initial graphs which display a certain kind
of (hidden) \textbf{selfsimilarity}, we expect that our coarse grained
graphs will change their character after a number of coarse graining
steps and become roughly isometric, i.e. structurally similar.
\end{conjecture}
This means that after, say, step $m$ of the coarse graining, the
operation 
\begin{equation}
\mathcal{K}:\left(G_{m-1},d_{m-1}\right)\to\left(G_{m},d_{m}\right)
\end{equation}
becomes a rough isometry. Note that operationally $\mathcal{K}$ obeys
the same protocols (contract cliques to new vertices etc.), but due
to the change of structure of our graphs, this operation now yields
a space that is roughly isometric to the previous one. 

At this point we should spell out a warning. In the introductory sections
we introduced the idea of phase cells of microstates which make up
the observable continuum limit space-time manifold S-T. We learned
previously that roughly isomorphic spaces have finite GH-distance
and thus have the same continuum limit. On the other hand, we cannot
expect that in the above sequence of coarse graining steps the roughly
isometric spaces converge to a limit on the microscopic level. What
we will observe however is that they belong to a joint macrocopic
continuum limit space (i.e. our classical space-time). To be more
precise we have the following. 
\begin{ob}
In general a sequence of roughly isometric spaces, while being structurally
similar, are not uniformly compact, thus there will not exist in the
generic case a GH-convergent subsequence. This can for example be
seen in the transition from a graph to its clique graph (cf. the numerical
estimates in \citep{RG}). While the two spaces are roughly isometric
the number of cliques may strongly increase so that the doubling property
is not fulfilled. On the other hand we studied very simple and regular
examples in section 4 of \citep{RG} and found real fixed points or
accumulation points. That is, it may be that both cases may happen
while we think, the latter case is not the generic one.
\end{ob}
Note in particular that on the discrete graph-level the GH-distance
is discrete. This implies that the smallest possible distances are
zero or one. This means that limit points on the discrete level, if
they exist at all, are attained after a finite number of steps and
remain stable under further coarse graining. The same happens with
possible accumulation points. These cases are illustrated by the above
mentioned examples.

In order to complete our coarse graining picture we now turn our attention
to the continuum limit of these spaces. In a first step we present
a theorem which turns out to be very useful in our context and which
we proved in \citep{Requ2}. 
\begin{theo}
Let $G_{1},G_{2}$ both have globally bounded vertex degree, let $G_{1}$
have uniform polynomial growth, and let $G_{1},G_{2}$ be quasi-isometric.
Then also $G_{2}$ has uniform polynomial growth. Thus in a sequence
of graphs $\left\{ \left(G_{i},d_{i}\right)\right\} ,\,i=0,\ldots N$
each derived from the previous one by a quasi-isometric coarse graining
process, each member has its own scaling limit as shown in theorem
\ref{thm:upg-converge}.
\end{theo}
This can be illustrated by means of the following graphical representation.
\[
\xymatrix{[G_{0}]_{\infty} & [\mathcal{K}(G_{0})]_{\infty} & [\mathcal{K}^{2}(G_{0})]_{\infty} & [\ldots]_{\infty} & [\mathcal{K}^{m}(G_{0})]_{\infty} & [\cdots]_{\infty} & [\mathcal{K}^{n}(G_{0})]_{\infty}\\
G_{0}\ar[r]^{\mathcal{K}}\ar[u]|-{\phi_{\lambda\rightarrow0}} & \mathcal{K}(G_{0})\ar[r]^{\mathcal{K}}\ar[u]|-{\phi_{\lambda\rightarrow0}} & \mathcal{K}^{2}(G_{0})\ar[r]^{\mathcal{K}}\ar[u]|-{\phi_{\lambda\rightarrow0}} & \cdots\ar[r]^{\mathcal{K}}\ar[u]|-{\phi_{\lambda\rightarrow0}} & \mathcal{K}^{m}(G_{0})\ar[r]^{\mathcal{K}}\ar[u]|-{\phi_{\lambda\rightarrow0}} & \cdots\ar[r]^{\mathcal{K}}\ar[u]|-{\phi_{\lambda\rightarrow0}} & \mathcal{K}^{n}(G_{0})\ar[u]|-{\phi_{\lambda\rightarrow0}}
}
\]
where $\phi_{\lambda\rightarrow0}$ stands for $\lim_{\lambda\rightarrow0}\phi_{\lambda}$,
and $\mathcal{K}^{i}(G_{0})=\underbrace{\mathcal{K}(\mathcal{K}(\ldots\mathcal{K}}_{i\,\textrm{times}}(G_{0})\ldots)=G_{i}$.
The lower horizontal chain contains the discrete spaces that could
be strictly quantum or semiclassical and each is derived from the
previous one by a coarse graining operation $\mathcal{K}$ that can
be a pure quasi-isometry or a rough isometry. The upper horizontal
chain contains the continuum limits of the discrete spaces in the
lower horizontal chain that are connected to their corresponding discrete
spaces by the recaling maps $\phi_{\lambda}$ for $\lambda\rightarrow0$.
We call the lower chain the coarse graining chain, and the upper one
the continuum limit chain.

Based on our previous discussions, two cases are now possible. Either
the graphs $G_{j},G_{j+1}=\mathcal{K}(G_{j})$, in the coarse graining
chain are connected by a pure quasi-isometry $\mathcal{K}$ in which
case they are structurally different, or by a rough-isometry $\mathcal{K}$
which means they are structurally similar since their GH-distance
is finite. We will now show what this means for their continuum limits.
We start by providing a simple lemma,
\begin{lemma}
For $X,X'$ being two metric spaces, we have 
\begin{equation}
d_{GH}(\lambda X,\lambda X')=\lambda\cdot d_{GH}(X,X').
\end{equation}
\end{lemma}
\begin{proof}
On $X,X',X\sqcup X'$ all metrics can be jointly scaled by a factor
$\lambda$. 
\end{proof}
Using this we can see that if in the above coarse graining chain of
graphs two members $G_{i},G_{j}$ have $d_{GH}(G_{i},G_{j})=\infty$,
then 
\begin{equation}
d_{GH}(G_{i,\infty},G_{j,\infty})=\infty
\end{equation}
i.e., their respective limit spaces also lie in different classes. 

Furthermore, we can see from the above lemma that if we enter the
regime where all the spaces $G_{l}$ are roughly isometric, that is,
if for two consecutive members $G_{i},G_{j}$ we have $d_{GH}(G_{i},G_{j})<\infty$,
then
\begin{equation}
d_{GH}(G_{i,\infty},G_{j,\infty})=0,
\end{equation}
and thus they have the same continuum limit. We argued above that
this joint continuum limit space represents the phase cell of microstates
which look classically or macroscopically the same.

We will conclude this section with a fundamental observation which
sheds some light on our continuum space-time on the various scales
of resolution. Furthermore it shows that our above observation concerning
the toy model of $\mathbb{Z}^{n}/\mathbb{R}^{n}$ was not accidental.
Let us take two elements $G_{i},G_{j}$ from the coarse graining chain
of graphs which are assumed to be only purely quasi-isometric and
each having a continuum limit $G_{i,\infty},G_{j,\infty}$ under the
limit $\lambda\to0$ or $l\to\infty$ of the scaling map 
\begin{equation}
\phi_{l}:\left(X,d_{X}\right)\longmapsto\left(X,\lambda d_{X}\right),\,\,\,\,\,\,\,\,\,\,\,\,\,\lambda=2^{-l}.
\end{equation}
Then,
\begin{theo}\label{Theo:DiffMetr-SameTop}
Under the above assumptions, the scaling limits $G_{i,\infty},G_{j,\infty}$
are homeomorphic topologically but carry different metrics. With the
help of the homeomorphism map, the limit spaces can then even chosen
to be the same topological space, but carrying different metrics.
That is, if our picture of S-T on the various scales is correct, the
various scales differ from each other with respect to the metric but
live on the same topological space. 
\end{theo}
For a proof, see appendix \ref{AppSec:Proofs} 

We would like to add a remark what this means physically. One should
note that the different metrics are not artificially imposed from
outside but result from the structural differences hidden in the deeper
layers of our space-time. That means, they may result for example
from the existence of short cuts or microscopic wormholes on certain
scales of our coarse graining process, a possibility we mention at
the end of the following section.

\section{Dimension\label{Dim}}

We mentioned the concept of \textbf{dimension} as a characteristic
of such discrete and irregular spaces like our graphs/networks in
the introduction. It was briefly remarked there why we chose our particular
notion, being guided mainly by purely physical motivations \citep{Requ3}.
We later realized that our concept is closely related to the notion
of \textbf{growth degree} in \textbf{geometric group theory} (see
section 4 in \citep{Requ2}).

It turned out that this notion has a lot of stability properties and
it is interesting to study its behavior under the \textbf{geometric
RG}. In \citep{Requ3} we studied two slightly different versions.
Here, we define 
\begin{equation}
D(G)=\lim_{r\to\infty}\frac{\log\beta(G,v_{i},r)}{\log r}\label{eq:dim-def}
\end{equation}
Note that in general only $\limsup$ and $\liminf$ of the right hand
side exist, but for convenience, we assume here that instead, its
ordinary limit does exist.
\begin{bem}
Note that in the cases we study, the value on the left hand side of
(\ref{eq:dim-def}) is independent of the reference vertex $v_{i}$
which is one of the stability properties of $D(G)$ as we mentioned
above.
\end{bem}
It is of tantamount importance as a characteristic of our limit or
continuum spaces, whether $D(G)$ is an integer or some non-integer
real value, which indicates the existence of a \textbf{fractal} limit
space. We begin by compiling some results we proved in \citep{Requ2}.
First we have the important result: 
\begin{theo}
If $G_{1},G_{2}$ are quasi-isometric, with globally bounded vertex
degree, then they have the same dimension (theorem 2.22 in \citep{Requ2}). 
\end{theo}
\begin{conclusion}
This implies that in our RG scenario, the dimension of the various
spaces remain constant under coarse-graining or scaling limit provided
the above assumption is fulfilled. This is then in particular the
case for the resulting continuum limits. 
\end{conclusion}
It follows that if one wants to have a changing dimension which depends
for example on the scale of spatial resolution one has to change these
assumptions.

Our macroscopic space-time is four-dimensional. It is a surprisingly
deep question in our context whether a general infinite graph/network
has an integer dimension. The whole section 3 of \citep{Requ2} was
devoted to this problem. This investigation culminated in the following
theorem: 
\begin{theo}
Let $G$ be a graph with locally finite vertex degree, being connected,
and \textbf{vertex transitive}, then its growth degree, and hence
its dimension, is an integer. The same holds then for graphs being
quasi-isometric to $G$. 
\end{theo}
That is, if we want to have an integer dimension, one way is to start
from such relatively homogeneous graphs (note that vertex transitivity
implies a constant vertex degree).

We want to briefly come back to the possibility of changing the dimension
under the RG, i.e. the picture that the individual scales of resolution
of our S-T may have their own dimensions (possibly non-integer ones
on more primordial scales). We started to discuss this possibility
in section 8 of \citep{RG} where we introduced \textbf{critical network
states}. These are states with both a local and a \textbf{translocal}
wiring structure. If we have to delete these translocal edges in the
course of reconstructing a smooth macroscopic space-time, we may leave
the class of quasi-isometries and hence our network dimension may
change. 

This idea was further explored in \citep{wormhole} where we developed
the concept of \textbf{wormhole spaces}. We argued that for example
the \textbf{BH-area law}, the \textbf{holographic principle} and \textbf{quantum
entanglement} can find a natural explanation in such a framework. 

\section{Conclusion\label{sec:Conclusion}}

In this work, we have laid out the basis of a novel approach to the
emergence of smooth spacetime from a discrete substratum based on
purely geometric notions. The approach resembles the Wilsonian renormalization
procedure. The starting point or the initial condition is a metric
space $\left(G_{0},d_{0}\right)$, where $G_{0}$ is an uncolored
graph\footnote{We postpone the treatment of colored graphs to a future work.}
and $d_{0}$ a graph metric defined on it. The renormalization procedure
contains two operations: a coarse graining operation
\begin{equation}
\mathcal{K}:G_{i}\to G_{i+1},
\end{equation}
that is a quasi-isometry map between discrete spaces, and a rescaling
map
\begin{equation}
\phi_{\lambda}:\left(G_{i},d_{i}\right)\to\left(G_{i},\lambda d_{i}\right)
\end{equation}
whose limit $\lambda\to0$ yields the continuum limit $\left(G_{i,\infty},d_{i,\infty}\right)$
of the discrete space $\left(G_{i},d_{i}\right)$,
\begin{equation}
\lim_{\lambda\to0}\phi_{\lambda}\left((G_{i},d_{i})\right)=\lim_{\lambda\to0}\left(G_{i},\lambda d_{i}\right)=\left(G_{i,\infty},d_{i,\infty}\right).
\end{equation}
Here the parameter $\lambda$, parametrizes the distance between the
points on the different length scales. Combining these two operations,
we represent again our findings in the following graphic.
\[
\xymatrix{[G_{0}]_{\infty} & [\mathcal{K}(G_{0})]_{\infty} & [\mathcal{K}^{2}(G_{0})]_{\infty} & [\ldots]_{\infty} & [\mathcal{K}^{m}(G_{0})]_{\infty} & [\cdots]_{\infty} & [\mathcal{K}^{n}(G_{0})]_{\infty}\\
G_{0}\ar[r]^{\mathcal{K}}\ar[u]|-{\phi_{\lambda\rightarrow0}} & \mathcal{K}(G_{0})\ar[r]^{\mathcal{K}}\ar[u]|-{\phi_{\lambda\rightarrow0}} & \mathcal{K}^{2}(G_{0})\ar[r]^{\mathcal{K}}\ar[u]|-{\phi_{\lambda\rightarrow0}} & \cdots\ar[r]^{\mathcal{K}}\ar[u]|-{\phi_{\lambda\rightarrow0}} & \mathcal{K}^{m}(G_{0})\ar[r]^{\mathcal{K}}\ar[u]|-{\phi_{\lambda\rightarrow0}} & \cdots\ar[r]^{\mathcal{K}}\ar[u]|-{\phi_{\lambda\rightarrow0}} & \mathcal{K}^{n}(G_{0})\ar[u]|-{\phi_{\lambda\rightarrow0}}
}
\]
where $[\mathcal{K}^{i}(G_{0})]_{\infty}=G_{i,\infty}$. The lower
chain consisting of discrete coarse grained spaces is called the coarse
graining chain while the upper one is the continuum limit chain.

Our coarse graining operations, $\mathcal{K}$, lie in the class of
quasi-isometries. We show that in the case, where two consecutive
members of the coarse chain $G_{i}$ and $G_{i+1}$ are related by
a true quasi-isometry, their continuum limits $G_{i,\infty},G_{i+1,\infty}$
carry different metrics, so their Gromov-Hausdorff (GH) distance,
$d_{GH}$, is infinite
\begin{equation}
d_{GH}\left(G_{i,\infty},G_{i+1,\infty}\right)=\infty.
\end{equation}
This distance is a measure of how much two spaces are (non)isometric,
and an infinite distance tells us that they are structurally distinct.
However, we show that, although they carry different metrics, they
are topologically homeomorphic, and can even be chosen to be the same
topological space. This implies that different levels of spacetime
will have different metric even if they are the same set. As to the
physical implications see the remarks at the end of section \ref{sec:Geo-Renorm}.

The coarse graining operation goes on until one of the two cases (or
both) happen: either $\mathcal{K}$ turns into a rough isometry, or
the sequence of graphs in the coarse graining chain reaches a stable
fixed point or a set of accumulation points already on the discrete
scale. 

In the former case, we show that the continuum limits of these roughly
isometric spaces are the same, in the sense that their GH-distance
is zero, $d_{GH}(G_{i,\infty},G_{i+1,\infty})=0$. This means that
they are isometric. We identify these roughly isometric spaces with
the phase cell of our state space which is the basin of attraction
for the corresponding continuum limit and which we associate (in the
optimal case) with our classical space-time.

The latter case happens if the sequence of coarse grained spaces are
uniformly compact, in which case, we can use the Gromov's compactness
theorem to show their convergence with respect to $d_{GH}$. We think
however that this possibility is not the generic one.

In any case, we show that the continuum limits are scale invariant,
in the sense that under the rescaling map
\begin{equation}
\phi_{\lambda}:\left(G_{i,\infty},d_{i,\infty}\right)\to\left(G_{i,\infty},\lambda d_{i,\infty}\right)
\end{equation}
their GH-distance is zero,
\begin{equation}
d_{GH}(X_{\infty},lX_{\infty})=0.
\end{equation}
Finally we present a relevant notion of dimension for these discrete
spaces and very briefly discuss its properties. It turns out that
under certain conditions, such as graphs of locally finite vertex
degree, being connected and vertex transitive, not only this dimension
is an integer, but is also stable (i.e. invariant) under both $\mathcal{K}$
and $\phi_{\lambda}$. However, if we deal with graphs/networks having
both a local and translocal wiring structure it may happen that the
coarse graining procedure is no longer local in the sense defined
above. In that case the dimension may become dependent on the coarse
graining scale so that each scale may have its own dimension (cf.
the remarks at the end of section \ref{Dim}).
\begin{acknowledgments}
The authors would like to thank the anonymous referee whose remarks and suggestions helped to improve the typescript.
S. R. would like to acknowledge the partial support of CONACyT Grant
No. 237351: Implicaciones F\'{i}sicas de la Estructura del Espaciotiempo.
He also would like to acknowledge the support of the PROMEP postdoctoral
fellowship (through UAM-I) and the grant from Sistema Nacional de
Investigadores of CONACyT.
\end{acknowledgments}

\appendix

\section{Some relevant definitions}
\begin{defi}
A \textbf{pseudometric} on a set $S$ is a map $d_{S}:S\rightarrow\mathbb{R}$
which has all the properties of the metric except the property that
$d_{S}(x_{1},x_{2})=0\Leftrightarrow x_{1}=x_{2}$ for all $x_{1},x_{2}\in S$.
\end{defi}
\begin{defi}
A metric space in which every sequence has a subsequence that converges
to a point in $M$ is called \textbf{sequentially compact}. For metric
spaces this is equivalent to the compactness defined via open covers.
\end{defi}
\begin{defi}
A metric space is called \textbf{proper} if all its closed balls,
$B(x,r),\,\,\,\forall x\in M$, are compact.
\end{defi}
\begin{defi}
A proper metric space is \textbf{locally compact} if every point has
a compact neighborhood.
\end{defi}
\begin{koro}
Proper spaces are locally compact, but the converse is not true in
general.
\end{koro}
\begin{defi}
A metric space $M$ is \textbf{complete} iff every Cauchy sequence
has a limit in $M$.
\end{defi}
\begin{defi}
A metric space $M$ is \textbf{bounded} if there exists some number
$r$, such that $d(x,y)\leq r,\,\forall x,y\in M$. The smallest possible
such $r$ is called the \textbf{diameter} of $M$.
\end{defi}
\begin{defi}
A metric space $M$ is \textbf{precompact} or \textbf{totally bounded}
if for every $r>0$ there exist finitely many open balls of radius
$r$ whose union covers $M$.
\end{defi}
\section{Proof of some of the theorems}\label{AppSec:Proofs}
\subsection*{Theorem \ref{Theo:Spencer}}
In a slightly different context a related result was already proved
in section VII of \citep{RG}. It is also contained in \citep{Lochmann}.
The order of a clique is bounded by $deg(G)+1$. The cliques, containing
a fixed vertex, $v_{0}$, are lying in its 1-neighborhood $B(v_{0},1)$,
they hence consist of subsets of $B(v_{0},1)$. Furthermore, no clique
is contained in another clique (due to maximality). They hence represent
a \textbf{Sperner system} (see e.g. \citep{Bollo2}) and it follows\footnote{In \citep{RG} we provided a simpler but cruder bound.}:
\begin{equation}
\#\,(cliques|v_{0}\in\;clique)\leq\begin{pmatrix}deg(G)+1\\
\lfloor(deg(G)+1)/2\rfloor
\end{pmatrix}=:Sper_{G}.
\end{equation}
We hence get that the vertex degree in $\mathcal{C}(G)$ is bounded
by
\begin{equation}
(deg(G)+1)\cdot(Sper_{G}-1).
\end{equation}
We now discuss the \textbf{rough isometry} between $G$ and $\mathcal{C}(G)$.
In a first step we define the map 
\begin{equation}
f:\,G\longrightarrow\mathcal{C}(G).
\end{equation}
For each vertex $v$ there exist cliques which contain $v$. We choose
one of them as $f(v)$. Let $C$ be an arbitrary clique in $\mathcal{C}(G)$.
It contains a vertex $v$ and we hence have
\begin{equation}
v\in f(v)\cap C\neq\emptyset.
\end{equation}
It follows that $\mathcal{C}(G)$ is contained in the one-neighborhood
of $f(G)$:
\begin{equation}
\mathcal{C}(G)\subset U_{1}(f(G)).
\end{equation}
Now let $v,v'$ be two arbitrary vertices in $G$. It exists a (\textbf{geodesic})
path, $\gamma$, from $v$ to $v'$ having a length $l(\gamma)=d(v,v')$.
Each pair, $(v_{j-1},v_{j})$, $j=1,\ldots,d(v,v')$, in the path
lies in a clique $C_{j}$. The consecutive cliques $C_{j},C_{j+1}$
have non-void overlap, i.e., they are connected by an edge in $\mathcal{C}(G)$.
We hence get a path in $\mathcal{C}(G)$ of length $d(v,v')-1$.

It follows that $f(v),f(v')$, which each can have at most distance
$1$ from the initial vertex $v$ or end vertex $v'$, have distance
at most $d(v,v')+1$. On the other hand, by the same token we can
conclude:
\begin{equation}
d(f(v),f(v')\geq d(v,v')-1,
\end{equation}
as a path, $\gamma'$ in $\mathcal{C}(G)$ between $f(v),f(v')$ of
length $l(\gamma')\leq d(v,v')-2$ yields a path of length $d(v,v')-1$
in $G$, which is a contradiction. We finally have:
\begin{equation}
d(v,v')-1\leq d(f(v),f(v'))\leq d(v,v')+1
\end{equation}
which proves our theorem.
\subsection*{Theorem \ref{Theo:DiffMetr-SameTop}}
The proof is given in several steps which may be useful for their
own sake. So let $G_{i},G_{j}$ be $(\lambda,\epsilon)$-quasi-isometric
with scaling limits $G_{i,\infty},G_{j,\infty}$. This implies that
there exists a map $f:G_{i}\to G_{j}$ with $d_{G_{j}}(y,f(G_{i})\leq\epsilon$
for all $y\in G_{j}$. We then have 
\begin{lemma}
$G_{j}$ and $f(G_{i})$ are roughly isometric. 
\end{lemma}
\begin{proof}
It holds 
\begin{equation}
d_{G_{j}}(y_{1},y_{2})\leq d_{G_{j}}(y_{1},f(x_{1})+d_{G_{j}}(f(x_{1}),f(x_{2}))+d_{G_{j}}(f(x_{2}),y_{2})
\end{equation}
for certain elements $x_{1},x_{2}$. The RHS is $\leq2\epsilon+d_{G_{j}}(f(x_{1}),f(x_{2}))$.
We now choose a map from $y\in G_{j}$ to a corresponding $f(x)$
for a suitable $x$ which defines an $\epsilon$-rough isometry. 
\end{proof}
\begin{koro}
It follows that $G_{j},f(G_{i})$ have the same continuum limit. 
\end{koro}
For the $x,x'$ which are mapped on the same $y$ under the quasi-isometry
$f$ we have 
\begin{equation}
0\geq\lambda^{-1}d_{G_{i}}(x,x')-\epsilon\quad\text{hence}\quad d_{G_{i}}(x,x')\leq\lambda\cdot\epsilon
\end{equation}
We choose a quasi-inverse $g$ to $f$ by selecting one element in
the preimage of $y$. We observe 
\begin{ob}
$G_{i}$ is also roughly isometric to $g(f(G_{i}))$. 
\end{ob}
\begin{conclusion}
We can restrict ourselves to the spaces $f(G_{i}),g(f(G_{i}))$. The
quasi-isometry defines a (bilipschitzian) equivalence 
\begin{equation}
\lambda^{-1}d_{G_{i}}(x,x')\leq d_{G_{j}}(f(x),f(x'))\leq\lambda d_{G_{i}}(x,x')
\end{equation}
 between the two spaces.
\end{conclusion}
Now we take the scaling limit on both sides and get 
\begin{equation}
2^{-l}\cdot\lambda^{-1}d_{G_{i}}(x,x')\leq2^{-l}\cdot d_{G_{j}}(f(x),f(x'))\leq2^{-l}\cdot\lambda d_{G_{i}}(x,x')\label{*}
\end{equation}
with corresponding bijective maps $f_{l}:g(f(G_{i}))\to f(G_{i})$.
We know that $2^{-l}G_{i},2^{-l}G_{j}$ converge in GH-sense to $G_{i,\infty},G_{j,\infty}$.
By theorem \ref{rough} and the following remarks we conclude that
there exist rough isometies between $2^{-l}G_{i},2^{-l}G_{j}$ and
$G_{i,\infty},G_{j,\infty}$. Equation (\ref{*}) shows that, in the
scaling limit, we get a continuous bijective map $f_{\infty}$ between
$G_{i,\infty},G_{j,\infty}$ with metrics $d_{i,\infty},d_{j,\infty}$
according to the bilipschitzian equivalence, described above. This
shows that the two limit spaces are homeomorphic. The following observation
concludes the proof. 
\begin{ob}
With the help of the map $f_{\infty}$ we can transfer the metric
structure from $G_{i,\infty}$ to $G_{j,\infty}$ and get two metrics
on the same space which are related to each other by the above bilipschitzian
equivalence which ultimately is a consequence of the quasi-isometry
we started from.
\end{ob}
A related result was proved in \citep{Lochmann} by different methods.

\bibliography{main}

\end{document}